\documentclass{article}%
\usepackage{amsfonts}
\usepackage{amsmath}
\usepackage{amssymb}
\usepackage{graphicx}%
\providecommand{\U}[1]{\protect\rule{.1in}{.1in}}
\newtheorem{theorem}{Theorem}
\newtheorem{acknowledgement}[theorem]{Acknowledgement}

\newtheorem{condition}[theorem]{Condition}

\newtheorem{lemma}[theorem]{Lemma}

\newtheorem{proposition}[theorem]{Proposition}
\newtheorem{remark}[theorem]{Remark}

\newenvironment{proof}[1][Proof]{\noindent\textbf{#1.} }{\ \rule{0.5em}{0.5em}}
\begin{document}

\title{Generalized R\'{e}nyi Entropy and Structure Detection of Complex Dynamical Systems}
\author{Gy\"{o}rgy STEINBRECHER\\Physics Department, Faculty of Science, University of Craiova,\\A. I. Cuza 13, 200585 Craiova, Romania \\Email: gyorgy.steinbrecher@gmail.com
\and Giorgio SONNINO\\Universit{\'{e}} Libre de Bruxelles (ULB), Department of Physics \\Bvd du Triomphe, Campus de la Plaine CP 231\\1050 Brussels, Belgium. Email: gsonnino@ulb.ac.be\\Royal Military School (RMS), Av. de la Renaissance 30, \\1000 Brussels, Belgium}
\maketitle

\begin{abstract}
We study the problem of detecting the structure of a complex dynamical system
described by a set of deterministic differential equation that contains a
Hamiltonian subsystem, without any information on the explicit form of
evolution laws. We suppose that initial conditions are random and the initial
conditions of the Hamiltonian subsystem are independent from the initial
conditions of the rest of the system. The single numerical information is the
probability density function of the system at one or several, finite number of
time instants. In the framework of the formalism of the generalized R\'{e}nyi
entropy we find necessary and sufficient conditions that the back reaction of
the Hamiltonian subsystem to the rest of the system is negligible.The results
can be easily generalized to the case of general, measure preserving subsystem.

\end{abstract}

\section{ Introduction}

Similar generalizations of the classical Shannon entropy \cite{shannon}
appeared independently both in mathematical \cite{Renyi1}, \cite{RenyiWT},
\cite{RenyiInfoAccumulation} as well as in the physical literature
\cite{Tsallis1}, \cite{Tsallis2}, \cite{TsallisBook}, \cite{TsallisGelMann}.
These generalizations contains the same functional \cite{SonninoSteinbrGRE},
that can be related to the Lebesgue space norm \cite{Rudin}, \cite{ReedSimon},
associated to the measure space in which the probability is defined. Observe
that both generalizations, by Alfred R\'{e}nyi and Constantino Tsallis, use
the minimal mathematical prerequisite necessary to define the generalized
entropy: the structure of measure space, and for instance no differentiable,
or further algebraic structures are assumed. Due to the similarity in their
definition, both formulation generate the same probability density function
(PDF), when maximal (generalized) entropy principle is used with constrained
optimization. This class of generalized entropies allows to study the case of
singular, normalizable PDF's, when the classical Shannon entropy is infinite
\cite{SonninoSteinbrGRE}. The naturalness of the R\'{e}nyi and Tsallis
entropies, from the point of view of the category theory was proven in
\cite{SGySASGCategory}. In the term physicists, this means that the functional
$N_{q}^{(1)}$ (see below) that appears both in the definition of the R\'{e}nyi
and Tsallis entropy, has nice properties. First it is multiplicative, in the
case of composed system whose components are not correlated, property that
translates in additivity of the R\'{e}nyi entropy (RE). Secondly, \ the
functional $N_{q}^{(1)}$is additive, in the case of composed system obtained
by the measure theoretic construction known under name "direct sum",
construction that appears in simplest case in \cite{Renyi1}. In the case \ of
PDF depending on many variables, it is possible to extend the previous
generalizations of the entropy, by using only one new fact, that remains in
the framework of the formalism of measure spaces: the product structure of the
measure space \cite{SonninoSteinbrGRE} . This new generalization extends the
geometrical interpretation of the RE. In this formalism the generalized
R\'{e}nyi entropy (GRE) is related to the norm ( generalized distance) of a
class of Banach spaces, a class of Lebesgue functional spaces with highly
anisotropic norm \cite{Besov}. It was proven that GRE has finite value also in
the case of large class of PDF's , whose Shannon or R\'{e}nyi entropies are
infinite \cite{SonninoSteinbrGRE}. The GRE\ is additive
\cite{SonninoSteinbrGRE} and the corresponding functional $N_{p,q}^{(2)}$ has
similar nice category theoretic properties as the functional $N_{q}^{(1)}$
\cite{SGySASGCategory}. In the study of dynamical systems (DS) driven by
external noise, modelling the anomalous transport in plasma \cite{Balescu1},
\cite{balescu3}, the GRE play the role of Liapunov functional
\cite{SonninoSteinbrGRE}. In this work we expose a new application of the GRE.

\ There are many situation when we have little information on physical object,
for instance in astrophysics. Despite we known that part of the system is
described by (classical, not quantum) Hamiltonians, or more generally, measure
preserving dynamics, there are interactions \ that are inaccessible to
observations. The situation under study is a complex dynamical system (DS)
$\Omega$, in a finite dimensional phase space and whose dynamics is known that
it is described by a set of ordinary differential equations. The \ DS $\Omega$
contains the interacting subsystems, $\Omega^{\prime}$ \ interacting with the
Hamiltonian subsystem $\Omega^{\prime\prime}$. \ The randomness, related to
the entropies, appears due to random initial conditions, and we suppose that
the initial conditions of the subsystems $\Omega^{\prime}$ and $\Omega
^{\prime\prime}$ are independent. In the our formalism the exact form of the
differential equations, of the Hamiltonian is not known. The accessible
information is the probability density function on the phase space of $\Omega
$, measured at a single or several, finite time instants. The exact
formulation is contained in Subsection \ref{markerSubsectTheDynamicalSystem}.
In Subsection \ref{markerSubsectionRenyiAndGeneralizedRenyi} we recall the
definition of the GRE and formulate the main result contained in the
Propositions \ref{markerPropositionInvariance},
\ref{markerPropositionSufficientCondition}, that means by computing the RE and
GRE, it is possible to decide about the existence of back reaction of the
Hamiltonian subsystem $\Omega^{\prime\prime}$ to the subsystem $\Omega
^{\prime}$. In the language of the ergodic theory, the fact that the subsystem
$\Omega^{\prime\prime}$ is driven by the subsystem $\Omega^{\prime}$, is
expressed by the statement: The DS $\Omega$ is the skew product of the
dynamical systems $\Omega^{\prime}$ with~$\Omega^{\prime\prime}$. The proof of
the results are \ contained in Section \ref{markersectionProofs}. The
applicability of the previous result to the discrete time dynamical systems,
when the continuos time differential equations are replaced by finite
difference recursion equations is \ treated in the Section
\ref{markerSectionDiscreteTimeDinamicalSystem}. More mathematical details are
in the Appendix. We warn the reader that the proof of the Lemma
\ref{markerLemmaRearangement} from the Appendix contains unproved, heuristic
assumptions. Its absolutely rigorous proof (the proof the convergence of the
finite dimensional approximations) requires a more elaborate mathematical
framework and restrictions, that is the subject of future studies.

\section{Statement of the problem and main result
\label{markerSectionStamentofthe problem}}

\subsection{Description of the dynamical system and main assumptions
\label{markerSubsectTheDynamicalSystem}}

Consider a composed physical system, defined in a phase space $\Omega$ , with
its subsystems $\Omega^{\prime}$ \ and $\Omega^{\prime\prime}$ \ \ with their
measure structures ,$(\Omega^{\prime},\mathcal{A}_{1},m_{1})$ , $(\Omega
^{\prime\prime},\mathcal{A}_{2},m_{2})$ , $\mathcal{A}_{1},~\mathcal{A}_{2}$
are the corresponding $\sigma-$algebras and $m_{1},~m_{2}$ are the
corresponding measures. \ We develop a general formalism, that also include
the case when the subsystem $\ \Omega^{\prime\prime}$ is Hamiltonian. The
global phase space is the direct product \ $\Omega=\Omega^{\prime}\times$
$\Omega^{\prime\prime}$\ , \ and its associated standard direct product
measure space structure \ $(\Omega^{\prime}\times\Omega^{\prime\prime
},\mathcal{A}_{1}\otimes\mathcal{A}_{2},m_{1}\otimes m_{2})$ . \ A typical
example is \ the case when $\Omega^{\prime}=\mathbb{R}^{N_{1}}$,
$\Omega^{\prime\prime}=\mathbb{R}^{N_{2}}$ , the $\sigma-$algebras
$\mathcal{A}_{1},$ $\mathcal{A}_{2}$ are generated by open subsets from
$\mathbb{R}^{N_{1}}$ respectively $\mathbb{R}^{N_{2}}$ and the measures
$m_{1}$ , $m_{2}$ are of the form%
\begin{align}
dm_{1}(\mathbf{x})  &  =\gamma_{1}(\mathbf{x})\prod\limits_{k=1}^{N_{1}}%
dx_{k}\label{10}\\
dm_{2}(\mathbf{y})  &  =\gamma_{2}(\mathbf{y})\prod\limits_{k=1}^{N_{2}}dy_{k}
\label{20}%
\end{align}
We consider the case when $\Omega^{\prime\prime}$ is isomorphic to
$\mathbb{R}^{N_{2}}$. Without loss of generality we can select a coordinate
system such that $\gamma_{2}(\mathbf{y})\equiv1$ and our result depends on
this selection. A more general case is when $\ \Omega^{\prime}$ is an
orientable differential manifold of dimension $N_{1}$ , and in this case
Eq.(\ref{10}) \ is expressed in some local coordinate system. \ Consider now
an evolution law on \ $\Omega^{\prime}\times$ $\Omega^{\prime\prime}$, that in
a local coordinate systems is of the form
\begin{align}
\frac{dx_{k}(t)}{dt}  &  =U_{k}(t,\mathbf{x});~k=\overline{1,N_{1}}%
\label{30}\\
\frac{dy_{j}(t)}{dt}  &  =V_{j}(t,\mathbf{y},\mathbf{x});~j=\overline{1,N_{2}}
\label{40}%
\end{align}
Consider the time dependent dependent vector field \ on $\Omega^{\prime\prime
}$%
\begin{equation}
\overline{V_{j}}^{(\mathbf{f})}(t,\mathbf{y}):=V_{j}(t,\mathbf{y}%
,\mathbf{f}(t));~j=\overline{1,N_{2}} \label{50}%
\end{equation}
where $\mathbf{f}(t)\mathbf{:=}\left(  f_{1}(t),\ldots,f_{N_{1}}(t)\right)  $
and the associated differential equation
\begin{equation}
\frac{dy_{j}(t)}{dt}=\overline{V_{j}}^{(\mathbf{f})}(t,\mathbf{y}%
,\mathbf{f}(t));~j=\overline{1,N_{2}} \label{51}%
\end{equation}
It generate an evolution map that preserve the measure $m_{2}$, irrespective
on the function $\mathbf{f}(t)\mathbf{:=}\left(  f_{1}(t),\ldots,f_{N_{1}%
}(t)\right)  $, iff:
\begin{equation}
\sum\limits_{j=1}^{N_{2}}\frac{\partial V_{j}(t,\mathbf{y},\mathbf{x}%
)}{\partial y_{j}}~=0 \label{60}%
\end{equation}
Observe that the important case of the Hamiltonian system \ in canonical
variables is recovered when $N_{2}=2d$, $\mathbf{y}=(\mathbf{q},\mathbf{p})$
and%
\begin{align*}
V_{j}(t,\mathbf{y},\mathbf{x})  &  =\frac{\partial}{\partial y_{d+j}}H\left(
t,\mathbf{y},\mathbf{x}\right)  ;~j=\overline{1,d}\\
V_{d+j}(t,\mathbf{y},\mathbf{x})  &  =-\frac{\partial}{\partial y_{j}}H\left(
t,\mathbf{y},\mathbf{x}\right)  ;~j=\overline{1,d}%
\end{align*}

Consider the situation when only the probability density function (PDF)\ of
the initial conditions associated to the Eqs.(\ref{30}, \ref{40}) \ is known.
In this case of the random initial conditions the evolution of the joint PDF
$\rho(t,\mathbf{x},\mathbf{y})$ is given by the following continuity equation
\begin{equation}
\frac{\partial\rho(t,\mathbf{x},\mathbf{y})}{\partial t}+\frac{1}{\gamma
_{1}(\mathbf{x})}\sum\limits_{k=1}^{N_{1}}\frac{\partial}{\partial x_{k}%
}\left[  \rho\gamma_{1}U_{k}\right]  +\sum\limits_{j=1}^{N_{2}}\frac{\partial
}{\partial y_{j}}\left[  \rho V_{j}\right]  =0 \label{70}%
\end{equation}
\ On the other hand the evolution of the subsystem $\Omega^{\prime}$ can be
studied independently. The evolution of the PDF in the phase space
$\Omega^{\prime}$ is described by the following continuity equation%
\begin{equation}
\frac{\partial\rho_{1}(t,\mathbf{x})}{\partial t}+\frac{1}{\gamma
_{1}(\mathbf{x})}\sum\limits_{k=1}^{N_{1}}\frac{\partial}{\partial x_{k}%
}\left[  \rho_{1}\gamma_{1}U_{k}\right]  =0 \label{90}%
\end{equation}

Our main assumption is the following

\begin{condition}
\label{markerConditionIndepInitCondi}The distribution of the random initial
conditions $\mathbf{x}(0)$, $\mathbf{y}(0)$ \ \ for Eqs.(\ref{30}, \ref{40})
are independent, \ that \ can be reformulated in the following condition on
the initial conditions for the solution of Eq.(\ref{70}), in term of the
solution $\rho_{1}(t,\mathbf{x})$ of Eq.(\ref{90})
\begin{equation}
\rho(0,\mathbf{x},\mathbf{y})=\rho_{1}(0,\mathbf{x})\rho_{2}(\mathbf{y})
\label{100}%
\end{equation}

\end{condition}

We can impose the normalization
\begin{equation}
\rho_{2}(\mathbf{y})=\int\limits_{\Omega_{1}}dm_{1}(\mathbf{x})\rho
(0,\mathbf{x},\mathbf{y}) \label{normalization}%
\end{equation}

\subsection{The R\'{e}nyi entropy, its generalization and the main
results\label{markerSubsectionRenyiAndGeneralizedRenyi}}

Starting from the reinterpretation of the RE in term of distance in Lebesgue
functional space, a generalization was introduced that preserve the additivity
in the case of composed system without correlation. In the our formalism
\cite{SonninoSteinbrGRE} the RE, $S_{R,q}$ associated to the subsystem
$\Omega^{\prime}$, described by the time dependent PDF $\rho_{1}%
(t,\mathbf{x})$ from Eq.(\ref{90}) is given by
\begin{equation}
S_{R,q}(t;\rho_{1}):=\frac{1}{1-q}\log N_{q}^{(1)}(t;\rho_{1}) \label{110}%
\end{equation}
where we denoted
\begin{equation}
N_{q}^{(1)}(t;\rho_{1}):=\int\limits_{\Omega^{\prime}}\left[  \rho
_{1}(t,\mathbf{x})\right]  ^{q}dm_{1}(\mathbf{x}) \label{120}%
\end{equation}
In the limit $q\rightarrow1$ the RE is equal with the Shannon-Boltzmann
entropy $S_{clasic}=-\int\limits_{\Omega^{\prime}}\rho_{1}(t,\mathbf{x}%
)\log\left[  \rho_{1}(t,\mathbf{x})\right]  dm_{1}(\mathbf{x})$ \ 

The solutions of Eqs.(\ref{70}, \ref{90}) are related by
\begin{equation}
\rho_{1}(t,\mathbf{x})=\int\limits_{\Omega^{\prime\prime}}\rho(t,\mathbf{x}%
,\mathbf{y})dm_{2}(\mathbf{y}) \label{130}%
\end{equation}
The version of interest of GRE, associated to the solution of Eq.(\ref{70}) is
given by
\begin{equation}
S_{p,q}(t;\rho):=\frac{1}{1-q}\log N_{p,q}(t;\rho) \label{140}%
\end{equation}
where we denoted%
\begin{equation}
N_{p,q}^{(2)}(t;\rho):={\int\limits_{\Omega^{\prime}}}dm_{1}(\mathbf{x)}%
\left[  {\int\limits_{\Omega^{\prime\prime}}}dm_{2}(\mathbf{y)}\left\vert
\rho(t,\mathbf{x},\mathbf{y})\right\vert ^{q}\right]  ^{p} \label{150}%
\end{equation}
In the case $p=1$ \ we obtain the RE and in the limit $q\rightarrow1$ the
limiting case of GRE, via RE, is the Shannon-Boltzmann entropy .

Define the following important functional:
\begin{equation}
I_{p,q}(t,\rho):=\left[  \frac{N_{p,q}^{(2)}(t;\rho)}{N_{pq}^{(1)}(t;\rho
_{1})}\right]  ^{1/p} \label{160}%
\end{equation}
where Eq.(\ref{130}) is assumed.\ We have the following

\begin{proposition}
\label{markerPropositionInvariance}Under the Condition
\ref{markerConditionIndepInitCondi} and previous assumptions, the functional
$I_{p,q}(t,\rho)$ associated to the solutions $\rho(t,\mathbf{x},\mathbf{y})$
and $\rho_{1}(t,\mathbf{x})$ of the Eqs.(\ref{70}, \ref{90}) has the following
invariance property:
\begin{equation}
I_{p,q}(t;\rho)\equiv{\int\limits_{\Omega^{\prime\prime}}}dm_{2}%
(\mathbf{y)}\left\vert \rho_{2}(\mathbf{y})\right\vert ^{q} \label{170}%
\end{equation}
and consequently its numerical value depends only on the initial distribution
function $\rho_{2}(\mathbf{y})$ of the measure preserving subsystem
$\Omega^{\prime\prime}$, and does not depend on the time $t,$ on the parameter
$p$ as well as on the function $\gamma_{1}(\mathbf{x})$ that define the
measure $dm_{1}(\mathbf{x})$ in Eq.(\ref{10}).
\end{proposition}

\begin{remark}
\label{markerRemarkInvariance}By using Eqs.(\ref{110}, \ref{120}, \ref{140},
\ref{150}), the previous Proposition \ref{markerPropositionInvariance} can be
reformulated in the term of RE of the PDF $\rho_{1},\rho_{2}$ and GRE of the
PDF $\rho$ as follows: the functional \
\begin{equation}
\log I_{p,q}(t;\rho)=\frac{1}{p}\left[  (1-q)S_{p,q}(t;\rho)-(1-pq)S_{R,pq}%
(t;\rho_{1})\right]  \label{172}%
\end{equation}
is independent on the time $t$, parameter $p$, the choice of the measure
$dm_{1}(\mathbf{x})$ and has the constant value
\begin{equation}
\log I_{p,q}(t;\rho)=(1-q)S_{R,q}(\rho_{2}) \label{173}%
\end{equation}

\end{remark}

The previous Proposition \ref{markerPropositionInvariance} or its equivalent
formulation Remark \ref{markerRemarkInvariance}, give necessary condition for
the absence of back reaction of the Hamiltonian subsystem $\Omega
^{\prime\prime}$ to the subsystem $\Omega^{\prime}$. In the following we
formulate a partial result in the reverse direction: by assuming the
invariance of $I_{p,q}$ we obtain a result about the structure of PDF similar
to Eq.(\ref{270}), a structure obtained assuming that there is no back
reaction of the \ Hamiltonian subsystem $\Omega^{\prime\prime}$ to
$\Omega^{\prime}$.

If the Condition \ref{markerConditionIndepInitCondi} with Eq.(\ref{100}) are
fulfilled, then we have the following result

\begin{proposition}
\label{markerPropositionSufficientCondition} Suppose that the functional
$I_{p,q}(t,\rho)$ from Eq.(\ref{170}) is independent on the time $t$,
parameter $p$ and the choice on the function $\gamma_{1}(\mathbf{x})$ that
define the measure $dm_{1}(\mathbf{x})$. Then there exists an map
$(t,\mathbf{x},\mathbf{y})\rightarrow T(t,\mathbf{x},\mathbf{y}),$ such that
for fixed $t,\mathbf{x}$ the map $\mathbf{y}\rightarrow T(t,\mathbf{x}%
,\mathbf{y})$ preserves the Lebesgue measure $dm_{2}(\mathbf{y})$ and we have
similar to Eq.(\ref{270})
\begin{equation}
\rho(t,\mathbf{x},\mathbf{y})=\rho_{1}\left(  t,\mathbf{x}\right)  \rho
_{2}\left(  T(t,\mathbf{x},\mathbf{y})\right)  \label{174}%
\end{equation}

\end{proposition}

\begin{remark}
\label{markerRemarkNoBackreaction}Observe that in the case of back reaction,
or equivalently, the case of completely coupled dynamical systems, the
evolution from $t=t_{1}$ to $t=0$ has the form $(\mathbf{x},\mathbf{y}%
)\rightarrow(h_{1}(t_{1},\mathbf{x},\mathbf{y}),h_{2}(t_{1},\mathbf{x}%
,\mathbf{y}))$ and the evolution of the PDF \ is more complicated, compared to
Eq.(\ref{174}):%
\begin{align}
\rho(t_{1},\mathbf{x},\mathbf{y})  &  =\rho(0,h_{1}(t_{1},\mathbf{x}%
,\mathbf{y}),h_{2}(t_{1},\mathbf{x},\mathbf{y}))K(t_{1},\mathbf{x}%
,\mathbf{y})=\label{175}\\
&  \rho_{1}(0,h_{1}(t_{1},\mathbf{x},\mathbf{y}))\rho_{2}(h_{2}(t_{1}%
,\mathbf{x},\mathbf{y}))K(t_{1},\mathbf{x},\mathbf{y})
\end{align}
where%
\[
K(t_{1},\mathbf{x},\mathbf{y})=\frac{\gamma(h_{1}(t_{1},\mathbf{x}%
,\mathbf{y}),h_{2}(t_{1},\mathbf{x},\mathbf{y}))}{\gamma(\mathbf{x}%
,\mathbf{y})}\frac{\partial(h_{1}(t_{1},\mathbf{x},\mathbf{y}),h_{2}%
(t_{1},\mathbf{x},\mathbf{y}))}{\partial(\mathbf{x},\mathbf{y})}%
\]

\end{remark}

\section{Proof of the results \label{markersectionProofs}}

\subsection{Proof of the Proposition \ref{markerPropositionInvariance}
\label{markerSectionProof}}

\begin{proof}
Denote by $g_{1}^{t,t_{0}}(\mathbf{x})$ the evolution map \cite{Arnold},
\cite{ArnoldAvez}, \cite{Aaronson} associated to the Eq.(\ref{30}): if
$\mathbf{x}(t)$ is a solution with $\mathbf{x}(t_{0})=\mathbf{x}_{0}$ then
$\mathbf{x}(t)=g_{1}^{t,~t_{0}}(\mathbf{x}_{0})$. Similarly we consider the
evolution map $(\mathbf{x},\mathbf{y})\rightarrow g^{t,t_{0}}(\mathbf{x}%
,\mathbf{y})$ associated to the \ system of differential equations
Eqs.(\ref{30}, \ref{40}), respectively let $\mathbf{y}\rightarrow
g_{\mathbf{f}}^{t,t_{0}}(\mathbf{y})$ the\emph{ measure preserving evolution
map attached to the equations Eqs.(\ref{50}, \ref{51}). \ }Then the evolution
of the PDF $\rho_{1}$ is given by the following equation%
\begin{equation}
\rho_{1}\left(  t_{1},g_{1}^{t_{1},~t_{0}}(\mathbf{x}_{0})\right)
\ \gamma_{1}\left(  g_{1}^{t_{1},~t_{0}}(\mathbf{x}_{0})\right)
\ \frac{\partial g_{1}^{t_{1},~t_{0}}(\mathbf{x}_{0})}{\partial\mathbf{x}_{0}%
}=\rho_{1}(t_{0},\mathbf{x}_{0})\gamma_{1}(\mathbf{x}_{0}) \label{180}%
\end{equation}
or \ setting $t_{1}=0,~\mathbf{x}_{0}=\mathbf{x}$ and $t_{0}=t$ \ , we obtain
\begin{equation}
\rho_{1}(t,\mathbf{x})=\frac{\ \gamma_{1}\left(  g_{1}^{0,~t}(\mathbf{x}%
)\right)  }{\gamma_{1}(\mathbf{x})}\frac{\partial g_{1}^{0,~t}(\mathbf{x}%
)}{\partial\mathbf{x}}\rho_{1}\left(  0,g_{1}^{0,~t}(\mathbf{x})\right)
\label{190}%
\end{equation}
We obtain the evolution law for the full PDF, if we decompose the map
$(\mathbf{x},\mathbf{y})\rightarrow g^{t,t_{0}}(\mathbf{x},\mathbf{y})$ as
follows
\begin{equation}
(\mathbf{x},\mathbf{y})\rightarrow g^{t_{1},t_{0}}(\mathbf{x},\mathbf{y}%
):=\left(  g_{1}^{t_{1},~t_{0}}(\mathbf{x}),~g_{2}^{t_{1},~t_{0}}%
(\mathbf{x},\mathbf{y})\right)  \label{200}%
\end{equation}
where $g_{1}^{t_{1},~t_{0}}$ is from Eq.(\ref{190}). It follows
\begin{equation}
\rho(t,\mathbf{x},\mathbf{y})=\rho\left(  0,g_{1}^{0,~t}(\mathbf{x}%
),g_{2}^{0,~t}(\mathbf{x},\mathbf{y})\right)  \frac{\gamma_{1}\left(
g_{1}^{0,~t}(\mathbf{x})\right)  }{\gamma_{1}(\mathbf{x})}J(t,\mathbf{x}%
,\mathbf{y}) \label{210}%
\end{equation}
where $J(t,\mathbf{x},\mathbf{y})$ is the Jacobian
\begin{equation}
J(t,\mathbf{x},\mathbf{y})=\frac{\partial\left(  g_{1}^{0,~t}(\mathbf{x}%
),~g_{2}^{0,~t}(\mathbf{x},\mathbf{y})\right)  }{\partial\ \left(
\mathbf{x},\mathbf{y}\right)  }=\frac{\partial\left(  g_{1}^{0,~t}%
(\mathbf{x})\right)  }{\partial\ \left(  \mathbf{x}\right)  }\frac
{\partial\left(  ~g_{2}^{0,~t}(\mathbf{x},\mathbf{y})\right)  }{\partial
\ \left(  \mathbf{y}\right)  } \label{220}%
\end{equation}
From the measure preserving property of the the map $\mathbf{y}\rightarrow
g_{\mathbf{f}}^{t,t_{0}}(\mathbf{y})$ results that for all $\mathbf{x}$ we
have (for details see the Appendix \ref{markerSectAppnedix})
\begin{equation}
\frac{\partial\left(  ~g_{2}^{0,~t}(\mathbf{x},\mathbf{y})\right)  }%
{\partial\ \left(  \mathbf{y}\right)  }=1;\ \forall\mathbf{x}\in\Omega
^{\prime} \label{240}%
\end{equation}
and the Eq.(\ref{210}) simplifies to the form
\begin{equation}
\rho(t,\mathbf{x},\mathbf{y})=\rho\left(  0,g_{1}^{0,~t}(\mathbf{x}%
),g_{2}^{0,~t}(\mathbf{x},\mathbf{y})\right)  \frac{\gamma_{1}\left(
g_{1}^{0,~t}(\mathbf{x})\right)  \ }{\gamma_{1}(\mathbf{x})\ }\frac
{\partial\left(  g_{1}^{0,~t}(\mathbf{x})\right)  }{\partial\ \left(
\mathbf{x}\right)  }\ \label{250}%
\end{equation}
and by using Condition \ref{markerConditionIndepInitCondi} it follows that
\begin{equation}
\rho(t,\mathbf{x},\mathbf{y})=\rho_{1}\left(  0,g_{1}^{0,~t}(\mathbf{x}%
)\right)  \rho_{2}\left(  g_{2}^{0,~t}(\mathbf{x},\mathbf{y})\right)
\frac{\gamma_{1}\left(  g_{1}^{0,~t}(\mathbf{x})\right)  \ }{\gamma
_{1}(\mathbf{x})\ }\frac{\partial\left(  g_{1}^{0,~t}(\mathbf{x})\right)
}{\partial\ \left(  \mathbf{x}\right)  } \label{260}%
\end{equation}
or by using Eqs.(\ref{190}, \ref{260}) we obtain
\begin{equation}
\rho(t,\mathbf{x},\mathbf{y})=\rho_{1}\left(  t,\mathbf{x}\right)  \rho
_{2}\left(  g_{2}^{0,~t}(\mathbf{x},\mathbf{y})\right)  \ \label{270}%
\end{equation}
We compute now $N_{p,q}^{(2)}(t;\rho)$ by using Eqs.(\ref{150}, \ref{270}) by
observing that from Eq.(\ref{240}) results that for all $\forall\mathbf{x}%
\in\Omega^{\prime}$ and any integrable function $F\left(  \mathbf{y}\right)  $
we have%
\begin{equation}
{\int\limits_{\Omega^{\prime\prime}}}dm_{2}(\mathbf{y)}F\left(  g_{2}%
^{0,~t}(\mathbf{x},\mathbf{y})\right)  ={\int\limits_{\Omega^{\prime\prime}}%
}dm_{2}(\mathbf{y)}F\left(  \mathbf{y}\right)  \label{280}%
\end{equation}
Consequently by using Eqs.(\ref{150}, \ref{270}) the rule Eq.(\ref{280}) and
the definition Eq.(\ref{120}) it follows
\begin{align}
N_{p,q}^{(2)}(t;\rho)  &  ={\int\limits_{\Omega^{\prime}}}dm_{1}%
(\mathbf{x)}\left[  \rho_{1}\left(  t,\mathbf{x}\right)  \right]  ^{pq}\left[
{\int\limits_{\Omega^{\prime\prime}}}dm_{2}(\mathbf{y)}\left[  \rho_{2}\left(
g_{2}^{0,~t}(\mathbf{x},\mathbf{y})\right)  \right]  ^{q}\right]
^{p}\label{290}\\
&  ={\int\limits_{\Omega^{\prime}}}dm_{1}(\mathbf{x)}\left[  \rho_{1}\left(
t,\mathbf{x}\right)  \right]  ^{pq}\left[  {\int\limits_{\Omega^{\prime\prime
}}}dm_{2}(\mathbf{y)}\left[  \rho_{2}(\mathbf{y})\right]  ^{q}\right]
^{p}\label{295}\\
&  =N_{q}^{(1)}(t;\rho_{1})\left[  {\int\limits_{\Omega^{\prime\prime}}}%
dm_{2}(\mathbf{y)}\left[  \rho_{2}(\mathbf{y})\right]  ^{q}\right]  ^{p}
\label{300}%
\end{align}

\end{proof}

which completes the proof of Proposition \ref{markerPropositionInvariance}

\subsection{Proof of Proposition \ref{markerPropositionSufficientCondition}}

Denote by $a$ the constant value%
\begin{equation}
a=\left[  \frac{N_{p,q}^{(2)}(t;\rho)}{N_{pq}^{(1)}(t;\rho_{1})}\right]
^{1/p} \label{p1}%
\end{equation}
By setting $t=0$ and from Condition \ref{markerConditionIndepInitCondi} we
obtain
\begin{equation}
a={\int\limits_{\Omega^{\prime\prime}}}dm_{2}(\mathbf{y)}\left\vert \rho
_{2}(\mathbf{y})\right\vert ^{q} \label{p2}%
\end{equation}
From Eq.(\ref{p1}) results%
\begin{equation}
{\int\limits_{\Omega^{\prime}}}dm_{1}(\mathbf{x)}\left\{  a^{p}\left[
\rho_{1}(t,\mathbf{x})\right]  ^{pq}-\left[  {\int\limits_{\Omega
^{\prime\prime}}}dm_{2}(\mathbf{y)}\left\vert \rho(t,\mathbf{x},\mathbf{y}%
)\right\vert ^{q}\right]  ^{p}\right\}  =0 \label{p3}%
\end{equation}
Due to the independence on the measure $dm_{1}$ it follows%
\begin{equation}
a\left[  \rho_{1}(t,\mathbf{x})\right]  ^{q}={\int\limits_{\Omega
^{\prime\prime}}}dm_{2}(\mathbf{y)}\left\vert \rho(t,\mathbf{x},\mathbf{y}%
)\right\vert ^{q} \label{p4}%
\end{equation}
From Eqs.(\ref{p2}, \ref{p4}) we obtain%
\begin{equation}
{\int\limits_{\Omega^{\prime\prime}}}dm_{2}(\mathbf{y)}\left\vert \rho
_{1}(t,\mathbf{x})\rho_{2}(\mathbf{y})\right\vert ^{q}={\int\limits_{\Omega
^{\prime\prime}}}dm_{2}(\mathbf{y)}\left\vert \rho(t,\mathbf{x},\mathbf{y}%
)\right\vert ^{q} \label{p5}%
\end{equation}
Fix for the moment the time $t$ and the variable $\mathbf{x}$ \ and we use the
Lemma \ref{markerLemmaRearangement} from the Appendix, with%
\begin{align*}
F(\mathbf{y})  &  =\rho(t,\mathbf{x},\mathbf{y})\\
G(\mathbf{y})  &  =\rho_{1}(t,\mathbf{x})\rho_{2}(\mathbf{y})
\end{align*}
Results that for all fixed $t,\mathbf{x}$ there exists a measure preserving
map%
\[
\Omega^{\prime\prime}\ni\mathbf{y\rightarrow}T(t,\mathbf{x,y})\in
\Omega^{\prime\prime}%
\]
$\ $such that%
\begin{equation}
\rho(t,\mathbf{x},\mathbf{y)=}\rho_{1}(\mathbf{x})\rho_{2}(T(t,\mathbf{x,y}))
\label{p6}%
\end{equation}
that completes the proof.

\bigskip

\bigskip

\section{Discrete time dynamical
systems.\label{markerSectionDiscreteTimeDinamicalSystem}}

By following the arguments in the proof of Propositions
\ref{markerPropositionInvariance}, \ref{markerPropositionSufficientCondition}
we observe that the conclusions remain valid if \ the dynamical systems are
described by finite difference evolution equations, instead of differential
equations Eqs.(\ref{30}, \ref{40}). We have the following evolutions on
$\Omega^{\prime}\times\Omega^{\prime\prime}$ where again $\Omega^{\prime
\prime}=\mathbb{R}^{N_{2}}$ :
\begin{align}
\mathbf{x}(t+\Delta t)  &  =\mathbf{X}\left(  t,\mathbf{x}(t)\right)
\label{dt1}\\
\mathbf{y}(t+\Delta t)  &  =\mathbf{Y}\left(  t,\mathbf{x}(t),\mathbf{y}%
(t)\right)  \label{dt2}%
\end{align}
where the second map preserves the Lebesgue measure ( it has unit Jacobian)
\[
\frac{\partial\mathbf{Y}(t,\mathbf{x},\mathbf{y})}{\partial\mathbf{y}}=1
\]
and in the case of discrete approximation of real physical systems the maps
are also orientation preserving. In the case when the map $\mathbf{y}%
\rightarrow\mathbf{Y}(t,\mathbf{x},\mathbf{y})$ is the finite time evolution
map of a Hamiltonian system, the matrix%
\[
\left[  \frac{\partial Y_{i}(t,\mathbf{x},\mathbf{y})}{\partial y_{j}}\right]
_{i,j=1,N_{2}}%
\]
is symplectic, but this property is not used in the proof. It is more easily
to construct integrators that preserve the volume in contrast to the
integrators that preserve the Poincar\'{e} invariants. If the PDF of the
distribution of the initial conditions fulfill the Condition
\ref{markerConditionIndepInitCondi} and $\rho(t,\mathbf{x},\mathbf{y})$ is the
joint PDF of the distribution at time $t$ generated this time by maps from
Eqs. Eqs.(\ref{dt1},\ref{dt2}), then \ Propositions
\ref{markerPropositionInvariance}, \ref{markerPropositionSufficientCondition}
are still valid. This is important in the studies when the evolution of the DS
is approximated \ by numerical integrators.

\section{Conclusions.\label{markerSectionConclusions}}

\ In the case of two of interacting dynamical systems, with independent random
initial conditions, when one system is Hamiltonian,it is possible to decide if
there is no back reaction of the Hamiltonian system to the remaining part of
the composed dynamical system. It is sufficient to compute the R\'{e}nyi
entropy and the generalized R\'{e}nyi entropy from the joint PDF at several
values of the parameters that specifies the generalized entropy, at some time
instants as well as different weight function associated to the measure in the
phase space. In the case of absence of back reaction the invariant defined in
Eq.(\ref{160}) does not depend on the parameter $p$, time $t$, the measure
$dm_{1}$.

\begin{acknowledgement}
Authors acknowledges M. van Schoor and D. van Eester of the Royal Military
School, Brussels. Gy${\ddot o}$rgy. Steinbreacher acknowledges J. H. Misguich and X. Garbet from IRFM,
C.E.A, Cadarache, France for useful discussions. Giorgio Sonnino is very grateful to Prof. Pasquale Nardone from the Universit{\'e} Libre de Bruxelles (U.L.B.) for his scientific suggestions.
\end{acknowledgement}

\bigskip

\section{\bigskip Appendix.}

\subsection{ Proof of Eq.(\ref{240})\label{markerSectAppnedix}}

Let $\mathbf{x}_{0}\in\Omega\,_{1}$ and denote by $\mathbf{X}(t,\mathbf{x}%
_{0})$ the unique solution of the Eq.(\ref{30}) with the property
$\mathbf{X}(0,\mathbf{x}_{0})=\mathbf{x}_{0}$. Then the joint solution
$\left(  \mathbf{x}(t),\mathbf{y}(t)\right)  $\ \ of the Eqs.(\ref{30},
\ref{40}) with the initial conditions $\left(  \mathbf{x}(0),\mathbf{y}%
(0)\right)  =\left(  \mathbf{x}_{0},\mathbf{y}_{0}\right)  $ is identical with%
\[
\left(  \mathbf{X}(t,\mathbf{x}_{0}),\mathbf{Y}(t,\mathbf{x}_{0}%
,\mathbf{y}_{0})\right)
\]
where $\mathbf{Y}(t,\mathbf{x}_{0},\mathbf{y}_{0})$ \ is the solution of the
Eqs.(\ref{50}, \ref{51}) with initial condition $\mathbf{Y}(0,\mathbf{x}%
_{0},\mathbf{y}_{0})=\mathbf{y}_{0}$ where we selected $\mathbf{f}%
(t)\equiv\mathbf{Y}(t,\mathbf{x}_{0},\mathbf{y}_{0})$%
\begin{align}
\frac{dY_{j}(t,\mathbf{x}_{0},\mathbf{y}_{0})}{dt}  &  =V_{j}(t,\mathbf{y}%
,\mathbf{Y}(t,\mathbf{x}_{0},\mathbf{y}_{0}));~j=\overline{1,N_{2}}%
\label{310}\\
\mathbf{Y}(t,\mathbf{x}_{0},\mathbf{y}_{0})  &  =\mathbf{y}_{0,} \label{320}%
\end{align}
Consequently the evolution map for the full system Eqs.(\ref{30}, \ref{40} )
is given by
\begin{equation}
\left(  \mathbf{x}_{0},\mathbf{y}_{0}\right)  \rightarrow\left(
\mathbf{X}(t,\mathbf{x}_{0}),\mathbf{Y}(t,\mathbf{x}_{0},\mathbf{y}%
_{0})\right)  =\left(  g_{1}^{t_{1},0}(\mathbf{x}_{0}),g_{2}^{t_{1}%
,0}(\mathbf{x}_{0},\mathbf{y}_{0})\right)  \label{330}%
\end{equation}
\ Let consider $\mathbf{x}_{0}$ fixed. From Eq.(\ref{60}) follows that the
evolution map obtained from Eq.(\ref{51}) , and in particular from
Eqs.(\ref{310})\ preserves the Lebesgue measure $dm_{2}(\mathbf{y})$%
\[
\frac{\partial\left(  ~\mathbf{Y}(t,\mathbf{x}_{0},\mathbf{y}_{0})\right)
}{\partial\ \left(  \mathbf{y}_{0}\right)  }=1;\ \forall\mathbf{x}\in
\Omega_{1}%
\]
which combined with Eq.(\ref{330}) completes the proof.

\subsection{Lemma on rearrangement}

We expose a simplified proof of the following Lemma

\begin{lemma}
\label{markerLemmaRearangement}Suppose that the functions $F(\mathbf{y})$,
$G(\mathbf{y})$ are non negative and in the complex neighborhood of $q_{0}>0$
the functions
\begin{align}
q  &  \rightarrow{\int\limits_{\Omega^{\prime\prime}}}dm_{2}(\mathbf{y)}%
\left[  F(\mathbf{y})\right]  ^{q}\label{r1}\\
q  &  \rightarrow{\int\limits_{\Omega^{\prime\prime}}}dm_{2}(\mathbf{y)}%
\left[  G(\mathbf{y})\right]  ^{q} \label{r2}%
\end{align}
are defined, are analytic and
\begin{equation}
{\int\limits_{\Omega^{\prime\prime}}}dm_{2}(\mathbf{y)}\left[  F(\mathbf{y}%
)\right]  ^{q}={\int\limits_{\Omega^{\prime\prime}}}dm_{2}(\mathbf{y)}\left[
G(\mathbf{y})\right]  ^{q} \label{r2a}%
\end{equation}
The there exists a measurable map $\Omega^{\prime\prime}\ni
\mathbf{y\rightarrow}T(\mathbf{y})\in\Omega^{\prime\prime}$ such that
\begin{equation}
F(\mathbf{y})=G\left(  T\left(  \mathbf{y}\right)  \right)  \label{r3}%
\end{equation}

\end{lemma}

\begin{proof}
We will approximate the map $T(\mathbf{y})$ by constructing a sequence of maps
$\mathbf{y\rightarrow}T_{n}(\mathbf{y})$. In this end \ first we select an
increasing sequence of finite hypercube subsets $\Omega_{n}\subset
\Omega^{\prime\prime}$, and in $\Omega_{n}$ we select a regular covering with
hypercube subdomains $\Omega_{n,k}$ , with $\ 1\leq k\leq M_{n}$ such that
there exists a vector $\mathbf{A}_{k,j}$ \ \ that translates $\Omega_{n,k}$ to
$\Omega_{n,j}$ and
\begin{align*}
\Omega_{n,k}  &  \subset\Omega_{n}\subset\Omega^{\prime\prime}\\
\cup_{k=1}^{M_{n}}\Omega_{n,k}  &  =\Omega_{n}\\
m_{2}(\Omega_{n,k})  &  =\frac{m_{2}(\Omega_{n})}{M_{n}}%
\end{align*}
\ Denote by $F_{k}=F(\mathbf{y}_{k})$ respectively by $G_{k}=G(\mathbf{y}%
_{k})$ the values of the functions \ $F(\mathbf{y}_{k}),G\left(
\mathbf{y}_{k}\right)  $ in some points \ \ $\mathbf{y}_{k}\in\Omega_{n,k}$.
The Eq.(\ref{r2a}) is approximated as follows%
\begin{equation}
\sum\limits_{k=1}^{M_{n}}\left(  F_{k}^{q}-G_{k}^{q}\right)  =0 \label{r4}%
\end{equation}
From Eq.(\ref{r4}) we prove that exists a permutation of the indices
$k\rightarrow P_{n}(k)$ such that
\begin{equation}
F_{k}=G_{P_{n}(k)} \label{r5}%
\end{equation}
Without losing generality, we may assume that $F_{k}>1$, \ $G_{k}>1\,$,
otherwise we multiply Eq.(\ref{r4}) with a suitable constant $c^{q}$.
Considering $q\rightarrow\infty$ \ in Eq.(\ref{r4})and denoting $k_{1}%
=\arg~\max~F_{k}$, observe that there exists an $k_{1}^{\prime}$ such that
$F_{k_{1}}=G_{k_{1}^{\prime}}$. Removing this term from the sum in
Eq.(\ref{r4}) we find the next value $\ k_{2}=\arg~\max\ F_{k};k\neq k_{1}$
and the corresponding value $k_{2}^{\prime}$ such that $F_{k_{2}}%
=G_{k_{2}^{\prime}}$ . Continuing the process we generate a permutation
$k_{j}\rightarrow k_{j}^{\prime}=P_{n}(k_{j})~$ such that $F_{k_{j}}%
=G_{k_{j}^{\prime}}$ . From permutation $P_{n}$ we generate the map
$\mathbf{y\rightarrow}T_{n}(\mathbf{y})$ such that to any $\mathbf{y\in}%
\Omega_{n,k}$ we associate $T_{n}(\mathbf{y)=y+A}_{k,P_{n}(k)}\in
\Omega_{n,P_{n}(k)}$. \ By increasing the number of subdomains $\Omega_{n,k}$
contained in $\Omega_{n}$ and increasing $\Omega_{n}$ \ such that $\cup
_{n}^{\infty}\Omega_{n}=\Omega^{\prime\prime}$ we obtain a sequence of maps
$T_{n}$ whose limit is the requested map $T$, that completes the proof.
\end{proof}

\end{document}